\documentclass[10pt,conference]{IEEEtran}
\IEEEoverridecommandlockouts
\usepackage{authblk}
\usepackage{cite}
\usepackage{amsmath,amssymb,amsfonts,dsfont}
\usepackage{algorithm}
\usepackage[noend]{algorithmic}
\usepackage{graphicx}
\usepackage{textcomp}
\usepackage{xcolor}
\usepackage{bm}
\usepackage{booktabs}
\usepackage{multirow}
\usepackage{lipsum}
\usepackage{url}
\usepackage{amsthm}

\ifCLASSOPTIONcompsoc
\usepackage[caption=false, font=normalsize, labelfont=sf, textfont=sf]{subfig}
\else
\usepackage[caption=false, font=footnotesize]{subfig}

\newtheorem{theorem}{Theorem}
\newtheorem{proposition}{Proposition}

\renewcommand{\algorithmicrequire}{\textbf{Input:}} 
\renewcommand{\algorithmicensure}{\textbf{Output:}}
\DeclareMathOperator*{\argmax}{\arg\max}

\newcommand{\mbbR}{\mathbb{R}} 
\newcommand{\mbbP}{\mathbb{P}}
\newcommand{\mbbE}{\mathbb{E}}

\newcommand{\bdx}{\boldsymbol{x}} 
\newcommand{\bdy}{\boldsymbol{y}} 
 
\newcommand{\bdT}{\boldsymbol{T}} 

\newcommand{\bda}{\boldsymbol{\alpha}} 
\newcommand{\bdb}{\boldsymbol{\beta}} 
\newcommand{\bdphi}{\boldsymbol{\phi}} 
\newcommand{\bdpsi}{\boldsymbol{\psi}}

\newcommand{\bdp}{\boldsymbol{p}}

\newcommand{\mcX}{\mathcal{X}}
\newcommand{\mcY}{\mathcal{Y}}
\newcommand{\mcC}{\mathcal{C}}

\newcommand{\mcM}{\mathcal{M}}
\newcommand{\mcP}{\mathcal{P}}


\begin{document}

\title
{
A Double Maximization Approach for Optimizing the LM Rate of Mismatched Decoding
\thanks{The first three authors contributed equally to this work and $\dag$ marked the corresponding author. This work was partially supported by National Key Research and Development Program of China (2018YFA0701603) and National Natural Science Foundation of China (12271289 and 62231022).}
}


\author[1]{Lingyi Chen}
\author[1]{Shitong Wu}
\author[1]{Xinwei Li}
\author[2]{Huihui Wu}
\author[1]{Hao Wu}
\author[3$\dag$]{Wenyi Zhang}
\affil[1]{Department of Mathematical Sciences, Tsinghua University, Beijing 100084, China}
\affil[2]{Yangtze Delta Region Institute (Huzhou), University of Electronic Science and Technology of China, 
\authorcr
Huzhou, Zhejiang, 313000, China.} 
\affil[3]{Department of Electronic Engineering and Information Science, 
\authorcr University of Science and Technology of China, Hefei, Anhui 230027, China 
\authorcr Email: wenyizha@ustc.edu.cn }

\maketitle

\begin{abstract}
An approach is established for maximizing 
the Lower bound on the Mismatch capacity 
(hereafter abbreviated as LM rate),
a key performance bound in mismatched decoding, by optimizing the channel input probability distribution. 
Under a fixed channel input probability distribution, the computation of the corresponding LM rate is a convex optimization problem. When optimizing the channel input probability distribution, however, the corresponding optimization problem adopts a max-min formulation, which is generally non-convex and is intractable with standard approaches. To solve this problem, a novel dual form of the LM rate is proposed, thereby transforming the max-min formulation into an equivalent double maximization formulation. This new formulation leads to a maximization problem setup wherein each individual optimization direction is convex. 
Consequently, an alternating maximization algorithm is established to solve the resultant maximization problem setup. Each step of the algorithm only involves a closed-form iteration, which is efficiently implemented with standard optimization procedures. Numerical experiments show the proposed approach for optimizing the LM rate leads to noticeable rate gains.

%
\end{abstract}
%

\section{Introduction}

%
The topic of mismatched decoding has aroused considerable attention since the 1970s \cite{omura1982coded}, due to its application to a myriad of practical scenarios, encompassing situations where channel knowledge is imperfect or where transceiver implementations are not fully optimized. 
Noteworthy scenarios include channels affected by uncertainties, such as fading in wireless communication systems \cite{Lapidoth1999Fading}, channels utilizing non-ideal transceiver hardware \cite{zhang2011general}, or channels employing constrained receiver structures \cite{salz1995Integer}. 
%
%
In such scenarios, it is common for the receiver to employ a prescribed decoding metric, which may not be matched to the actual channel transition law.
%
%
As a result, extensive research has been conducted regarding the fundamental principles of mismatched decoding; see, e.g., \cite{Merhav1994Mismatched,Lapidoth1998Reliable,lapidoth1996mismatched,Ganti2000Mismatched,Scarlett2020Information} and references therein. 

The mismatch capacity, which characterizes the supreme of achievable information rates under a prescribed decoding metric, has been introduced to evaluated the ultimate performance limit of mismatched decoding \cite{Lapidoth1998Reliable,Scarlett2020Information}. When the decoding metric is matched to the channel transition law, this reduces to the familiar channel capacity.
For the general mismatched case, to date, the mismatch capacity remains an open problem \cite{csiszar1995Channel}. 
%
%
%
%
By constructing different codebook ensembles, several lower bounds of the mismatch capacity have been developed. 
%
These include the generalized mutual information (GMI) based on independent and identically distributed (i.i.d.) random codebooks \cite{Kaplan1993Information}, the LM (“Lower [bound on the] Mismatch [capacity]”) rate based on constant-composition random codebooks \cite{csiszar1981graph}, and several improvements based on the GMI and the LM rate combined with more sophisticated techniques like superposition \cite{Scarlett2020Information}. 
Furthermore, we can maximize these lower bounds by optimizing the channel input probability distribution, leading to tighter lower bounds on the mismatch capacity. 
In this work, we study the maximized LM rate over all feasible channel input probability distributions, denoted as $C_\mathrm{LM}$, considering the fact that under the same channel input probability distribution, the LM rate is, 
in general, a better lower bound than the GMI.

However, the computation of $C_\mathrm{LM}$ is challenging. 
With a prescribed channel input probability distribution, the computation of the GMI and the LM rate can be deduced into convex optimization problems which can be readily solved by solvers such as the CVX \cite{boyd2004convex}. 
The computation of $C_\mathrm{LM}$, instead, is a max-min optimization problem, which turns out to be generally non-convex \cite{Scarlett2020Information}. 
This renders directly invoking convex optimization solvers infeasible. 
%
%
%
%
Moreover, due to the max-min optimization problem formulation and additional constraints relating the channel input and output probability distributions via the channel transition law, 
the computation of $C_\mathrm{LM}$ cannot be directly cast as the standard optimization problems or an entropy regularized optimal transport problem like those in \cite{ye2022optimal,wu2022communication}, and cannot be directly solved via gradient or alternating procedures, such as the Sinkhorn algorithm \cite{2013sinkhorn}.
The above discussions may explain why we have not seen any work on computing $C_\mathrm{LM}$ so far.

In this paper, we propose a novel approach for computing $C_\mathrm{LM}$ over a discrete memoryless channel (DMC). 
To address the difficulty due to the max-min problem formulation, we propose a new dual form of the LM rate. This transforms the max-min optimization problem into a double maximization problem, which enables us to develop an alternating maximization algorithm with guarantee of local convergence.
%
%
%
%
Moreover, we construct a variable transform and then propose a maximization model for computing $C_\mathrm{LM}$ in which each optimization direction is convex. This leads to an alternating maximization algorithm, termed Alternating Double Maximization (ADM) algorithm, for solving the proposed maximization model, whose each step only involves a closed-form iteration for alternating ascent.
%
%
%
Numerical experiments show that for Gaussian channels with IQ imbalance 
under QPSK, 16QAM, 64QAM and 256QAM constellations, the proposed algorithm is efficient and leads to noticeable rate gains when optimizing the channel input probability distribution.

\section{Problem Formulation}\label{sec_problem}

%
We consider a DMC with transition law (i.e., conditional probability distribution) $W(y|x)$ over the channel input alphabet $\mcX=\{x_1,\cdots,x_M\}$ and the channel output alphabet $\mcY=\{y_1,\cdots,y_N\}$.
Given a channel input probability distribution $P_{X}\in\mcP(\mcX)$, the channel input-output joint probability distribution $P_{XY}\in\mcP(\mcX\times\mcY)$ and the channel output distribution $P_{Y}\in\mcP(\mcY)$ 
are then induced by the transition law \cite{Ganti2000Mismatched}.

An encoding scheme is represented by a codebook $\mcC_{n}$ consisting of $2^{nR}$ length-$n$ sequences $\{\bdx^{n}(m)\}_{m=1}^{2^{nR}}$.
%
The encoder assigns a codeword $\bdx^{n}(m)\in\mcX^{n}$ to each message index $m$ uniformly selected from the message set $\mcM=\{1, 2, \ldots, 2^{nR}\}$.
The decoder assigns an estimate $\hat{m}\in\mcM$ to each received channel output sequence $\bdy^{n}\in\mcY^{n}$ according to the following prescribed decoding rule:
\begin{equation*}
    \hat{m} = \argmax_{j\in\mathcal{M}}\prod_{i=1}^{n}q(x_{i}(j),y_{i}),
\end{equation*}
where $q: \mcX\times\mcY\to\mbbR$ is called the decoding metric, and when $q(x, y)$ is not proportional to $W(y|x)$, it is called a \textit{mismatched decoding metric}.

Given a pair of encoder and decoder, the associated error probability is given by $P_{e}^{(n)}=\mbbP[\hat{m}\neq m]$ where the probability is defined with respect to the randomness of the message and the DMC.
A rate $R$ is said to be achievable if $\lim_{n\rightarrow\infty}P_{e}^{(n)}=0$ under the 
decoding rule.
Given a probability distribution $Q_X\in\mcP(\mcX)$, if the codebook $\mcC_{n}$ is constructed in such a way that each codeword $\bdx^{n}(m)$ has its composition (a.k.a. type) fixed as $Q_X$ \cite{csiszar1981graph}, and all codewords are independent, then the following so-called LM rate \cite{Merhav1994Mismatched}:
\begin{equation}\label{LM_def}
I_{\mathrm{LM}}\left(Q_X\right)=\min_{\substack{\widetilde{P}_{X Y} \in \mathcal{P}(\mathcal{X} \times \mathcal{Y}): \widetilde{P}_X=Q_X, \widetilde{P}_Y=P_Y \\ \mathbb{E}_{\widetilde{P}}[\log q(X, Y)] \geq \mathbb{E}_P[\log q(X, Y)]}} I_{\widetilde{P}}(X ; Y),
\end{equation}
is achievable, where the subscript $\widetilde{P}$ indicates that the corresponding expectation and mutual information are with respect to the  auxiliary joint probability distribution $\widetilde{P}_{X Y}$, and $P_Y$ 
is induced by marginalizing the joint probability distribution 
$P_{XY}(x,y) = W(y|x)Q_X(x)$.
%
%
%

Consequently, when optimizing over $Q_X \in \mcP(\mcX)$, we can introduce
\begin{equation}\label{LMC_def}
    C_\mathrm{LM} = \max_{Q_X\in\mcP(\mcX)}I_{\mathrm{LM}}\left(Q_X\right),
\end{equation}
as the optimized LM rate, 
which is the objective of study in our work.
%
%
%
In practice, the set $\mcP(\mcX)$ can impose certain constraint on the channel input. 
A typical constraint is an average power constraint $\Gamma$ like 
$\mbbE \left[X^2 \right]\leq \Gamma$
as considered in the sequel.
%
%

%
%

For a given $Q_X$, $I_\mathrm{LM}(Q_X)$ as given by \eqref{LM_def} is a convex optimization problem, and thus can be readily computed by solvers like CVX \cite{boyd2004convex}.
When optimizing over $Q_X$, it is evident that the resulting problem of computing $C_\mathrm{LM}$ adopts a max-min form, and becomes a non-convex optimization problem \cite{Scarlett2020Information}.
As a result, computing $C_\mathrm{LM}$ is not amenable to standard algorithms or solvers, due to the intrinsic max-min structure and the non-convex nature of objective function.
%

%
%
%

\section{Double Maximization Model and Alternating Maximization Algorithm} \label{sec_main}

%
The key to dealing with the challenges in the computation of $C_{\mathrm{LM}}$ is the introduction of a novel dual form of $I_{\mathrm{LM}}$.
This converts the max-min problem into a double maximization model, and decouples the constraints to produce simpler alternatives.
%
%
%
Subsequently, we design an alternating maximization algorithm to solve the transformed maximization problem. 
%
%
It is worth emphasizing that our approach ensures each step in the algorithm only involves a closed-form iteration, which can be efficiently handled by standard optimization procedures.
The alternating maximization algorithm exhibits local convergence behavior, which will also be confirmed by our numerical experiment in Section \ref{sec_numerical}.

\subsection{Double Maximization Model}



In fact, prior research has explored some dual formulations of the LM rate problem. 
For instance, a recent work by \cite{ye2022optimal} adopted an optimal transport (OT) approach to solve the LM rate problem whose dual form was obtained by analyzing its Lagrange function, featuring a kernel matrix multiplication structure. 
Besides, some other dual forms and their equivalence have been 
summarized in \cite{Scarlett2020Information}.

However, when computing $C_\mathrm{LM}$, the 
channel input probability distribution is treated as parameters, and 
is related to the channel output probability distribution
via the channel transition law. 
This situation is much more complicated than the classical LM rate problem. 
This aspect distinguishes it from the previous LM rate problem wherein these parameters are prescribed.
Furthermore, in algorithm design, it is essential to employ appropriate dual forms capable of producing closed-form iterations to 
improve algorithm efficiency.
In this regard, the dual forms proposed in earlier studies are not suitable for direct application to compute $C_\mathrm{LM}$.

To solve the aforementioned difficulties, we propose a novel dual form,
%
%
which takes into account the coupling effect between the channel input and output probability distributions, 
and exhibits the desirable property of generating closed-form iterations when designing algorithms.
To simplify the notations, we introduce
%
\begin{equation} \label{annotation}
\begin{aligned}
    &\gamma_{ij}=\widetilde{P}_{X Y}(x_i,y_j),\quad p_{i}={Q}_X(x_i),\quad q_{j}={P}_Y(y_j), \\
    &d_{ij}=-\log {q(x_i,y_j)},\quad s_{ij}=W(y_j | x_i), 
\end{aligned}
\end{equation}
where $q_j = \sum_{i=1}^M s_{ij} p_i$ according to the channel transition law.
%
%
%
%
Then, 
the following result presents the proposed dual form.
%
%
%
%
%
\begin{theorem}
For a fixed 
$Q_X\in\mcP(\mcX)$, a dual form of the LM rate problem \eqref{LM_def} can be written as:
    \begin{equation}
    \begin{split}
        &\max_{\substack{\bdphi, \bdpsi, \\ \zeta \geq 0}} 
        \left(
        -\sum_{i=1}^M \sum_{j=1}^N \phi_ie^{-\zeta d_{ij}}\psi_j - \sum_{i=1}^M p_i \log p_i \right.\\
        &\quad\left. -\!\! \sum_{j=1}^N \left(\sum_{i=1}^M s_{ij} p_i \right) \log \left(\sum_{i=1}^M s_{ij} p_i \right) 
        \!\!+\!\! \sum_{i=1}^M p_i \log \phi_i 
        \right.
        \\
        &\quad\left.
        +\!\! \sum_{j=1}^N \left(\sum_{i=1}^M s_{ij} p_i \right) \log \psi_j 
        \!+\! 1 \!-\! \zeta \sum_{i=1}^M \sum_{j=1}^N d_{ij}s_{ij}p_i
        \right),
    \end{split}
    \label{LM_dual}
    \end{equation}
where $\bdphi\in\mbbR^{M}, \bdpsi\in\mbbR^{N}$ and $\zeta\in\mbbR^{+}$.
\end{theorem}
\begin{proof}
%
%
We construct the Lagrangian of $\eqref{LM_def}$ by introducing the dual variables $\bm \alpha \in \mathbb{R}^M$, $\bm \beta \in \mathbb{R}^N$, $\zeta \in \mathbb{R}^+$:
\begin{equation*} 
    \begin{aligned}
        &\mathcal{L}_{\mathrm{LM}}(\bm\gamma; \bm\alpha ,\bm\beta ,\zeta) = 
        \sum_{i=1}^M \sum_{j=1}^N \gamma_{ij}\log \gamma_{ij} 
        - \sum_{i=1}^M p_i \log p_i \\
        & -\!\! \sum_{j=1}^N q_j \log q_j + \zeta \left( \sum_{i=1}^M \sum_{j=1}^N d_{ij}\gamma_{ij} \!-\!\!\sum_{i=1}^M \sum_{j=1}^N d_{ij}s_{ij}p_i  \right) \\
        & + \sum_{i=1}^M \alpha_i \left( \sum_{j=1}^N \gamma_{ij} - p_i \right) + \sum_{j=1}^N \beta_j\left( \sum_{i=1}^M \gamma_{ij} - q_j\right).
    \end{aligned}
    \label{lm_lagarange}
\end{equation*}
%
%
By taking the derivative of $\mathcal{L}_{\mathrm{LM}}(\gamma; \bda, \bdb, \lambda)$ with respect to $\gamma_{ij}$, 
we can have the optimal solution $\gamma_{ij}^* = e^{-\alpha_{i}-\beta_{j}-\zeta d_{ij}-1}.$

Next, $\phi_i = e^{-\alpha_i-1/2}$, $\psi_j = e^{-\beta_j-1/2}$ are denoted for short. Then, the dual problem is written as 
\begin{equation*}
    \max_{\bdphi, \bdpsi, \zeta}~~ g_{\mathrm{LM}}(\bdphi, \bdpsi, \zeta)
\end{equation*}
where $g_{\mathrm{LM}}(\bdphi, \bdpsi, \zeta)$ is obtained by substituting the optimal solution $\gamma_{ij}^*$ into the Lagrangian $\mathcal{L}_{\mathrm{LM}}(\gamma; \bda, \bdb, \lambda)$, i.e., 
\begin{small}
\begin{align*}
    & g_{\mathrm{LM}} (\bdphi ,\bdpsi ,\zeta) = 1-\sum_{i=1}^M \sum_{j=1}^N \phi_i e^{-\zeta d_{ij}} \psi_j - \zeta \sum_{i=1}^M \sum_{j=1}^N d_{ij}s_{ij}p_i\\
    & - \sum_{i=1}^M p_i \log p_i - \sum_{j=1}^N q_j \log q_j  + \sum_{i=1}^M p_i \log \phi_i+ \sum_{j=1}^N q_j \log \psi_j.
\end{align*}
\end{small}
Noting the marginal condition $q_{j}=\sum_{i=1}^M p_i s_{ij}$, we obtain the dual expression $\eqref{LM_dual}$ by substituting it into $g_{\mathrm{LM}}(\bdphi ,\bdpsi ,\zeta)$.
\end{proof}


\begin{proposition} \label{prop:equi}
    The dual form presented in \eqref{LM_dual} is equivalent to the one introduced in \cite[p. 19]{Scarlett2020Information}.
\end{proposition}
\begin{proof} 
    The proof is mainly built on the properties of conditional probability distributions. 
    Details are presented in the appendix. 
\end{proof}

While these two dual forms are equivalent, the dual form in \cite{Scarlett2020Information} is inconvenient to be applied for computing $C_{\mathrm{LM}}$ directly due to the lack of closed-form iterations. 
Next, we will present our proposed approach
based on the dual form \eqref{LM_dual}, which guarantees closed-form iterations.
\vspace{+.05in}

By substituting the dual form \eqref{LM_dual}, the max-min optimization problem \eqref{LMC_def} can be converted into the following double maximization model:
\begin{equation} \label{double_max}
    \begin{aligned}
        \max_{\bdp}& \max_{\substack{\bdphi, {\bdpsi}, \\ \zeta \geq 0}} ~~
        \left(
        1
        -\sum_{i=1}^M \sum_{j=1}^N \phi_ie^{-\zeta d_{ij}}\psi_j 
        - \sum_{i=1}^M p_i \log p_i 
        \right.\\
        & \left. 
        -\!\! \sum_{j=1}^N \left(\sum_{i=1}^M s_{ij} p_i \right) \log \left(\sum_{i=1}^M s_{ij} p_i \right) 
        \!\!
        +\!\! \sum_{i=1}^M p_i \log \phi_i 
        \right.
        \\
        & \left.
        +\!\! \sum_{j=1}^N \left(\sum_{i=1}^M s_{ij} p_i \right) \log \psi_j 
        \!-\! \zeta \sum_{i=1}^M \sum_{j=1}^N d_{ij}s_{ij}p_i
        \right)
        .
    \end{aligned}
\end{equation}

%

\subsection{Alternating Maximization Algorithm}

It should be noted that utilizing the model \eqref{double_max} for computing $C_\mathrm{LM}$ directly results in challenges when updating the variable $\bdp$ with a closed-form solution. 
%
This challenge arises due to the complexity introduced by the entropy term $\sum_{j=1}^N \left(\sum_{i=1}^M s_{ij} p_i \right) \log \left(\sum_{i=1}^M s_{ij} p_i \right)$ in \eqref{double_max}.
Hence, to ensure a closed-form solution when updating $\bdp$, we introduce
an elegant variable transform 
$\widetilde{\psi_j} \triangleq \frac{\psi_j}{q_j}$ in the following model.
Using the notations in \eqref{annotation} and the dual form in \eqref{LM_dual}, we further propose the following model for computing $C_\mathrm{LM}$.
\begin{proposition}
The optimized LM rate $C_\mathrm{LM}$ in \eqref{LMC_def} with the average power constraint $\Gamma$ is equivalent to solving the following maximization problem: 
\begin{equation}
    \begin{aligned}
        \max_{\substack{\bdp, \bdphi, \widetilde{\bdpsi}, \\ \zeta \geq 0}} \quad & - \sum_{i=1}^M p_i \log p_i + \sum_{i = 1}^M p_i \log T_i(\bdphi, \widetilde{\bdpsi}, \zeta) + 1 \\
        \mathrm{s.t.}\quad &  \sum_{i = 1}^M p_i  = 1, \quad \sum_{i = 1}^M p_i \|x_i\|^2 \leq \Gamma, 
    \end{aligned}
    \label{double_max_new}
\end{equation}
where the function $T_i(\bdphi, \widetilde{\bdpsi}, \zeta)$ 
is defined as
%
\begin{equation*}
    T_i
    \triangleq \phi_i \exp
    \left(\sum\limits_{j=1}^N s_{ij}
    \left[  \!-\!\!\!\sum\limits_{k=1}^M\!\! \phi_k e^{-\!\zeta d_{kj}}\widetilde{\psi_j}
    \!\!+\!\! \log\widetilde{\psi_j}
    \!\!-\!\!  \zeta d_{ij}  \right]
    \right).
\end{equation*}
%
\end{proposition}
\begin{proof}
    By substituting $T_i(\bdphi, \widetilde{\bdpsi}, \zeta)$ into \eqref{double_max} and extracting the coefficients $p_{i}$, we obtain the maximization model \eqref{double_max_new}.
\end{proof}

Based on the above maximization model, the key idea is to update $\bdp$ and $\bdphi, \widetilde{\bdpsi}, \zeta$ in an alternating ascending way. 
%
%
%
%
It is worth noting that,
although the proposed steps correspond to solving a sub-problem 
with constraint when updating $\bdp$ due to the introduction of the power constraint, we can still update $\bdp$ in closed-form through the proposed model \eqref{double_max_new}.
%

\subsubsection{Fix $\bdT(\bdphi, \widetilde{\bdpsi}, \zeta)$ and update $\bdp$}

For fixed $\bdphi, \widetilde{\bdpsi}, \zeta$, the maximization problem \eqref{double_max_new} can be regarded as an optimization problem with respect to $\bdp$ and its Lagrangian is given by: 
\begin{equation}
    \begin{aligned}
        \mathcal{L}(\bdp; & \lambda, \eta) = 
         - \sum_{i=1}^M p_i \log p_i + \sum_{i = 1}^M p_i \log T_i  +1 \\
        & - \lambda\left( \sum_{i = 1}^M p_i \|x_i\|^2 - \Gamma\right)
        - \eta\left( \sum_{i = 1}^M p_i - 1 \right),
    \end{aligned}
    \label{lmc_lagarange}
\end{equation}
where $\lambda \in \mathbb{R}^+$ and $\eta \in \mathbb{R}$ are dual variables.

By taking the partial derivative of the Lagrangian $\mathcal{L}(\bdp; \lambda, \eta)$ with respect to $\bdp$, we could obtain
the optimal solution $$p_i^{*} = T_i e^{- \lambda \|x_i\|^2} e^{-1 - \eta}.$$
%
Noting that $\sum_{i=1}^M p_i=1$, we can update $\bdp$ by
\begin{equation}
    p_i = \frac{T_i e^{- \lambda \|x_i\|^2}}{\sum_{i^{'} = 1}^M T_{i^{'}} e^{- \lambda \|x_{i^{'}}\|^2}}
    , \quad i = 1,\cdots, M,  
    \label{update_p}
\end{equation}
where the multiplier $\lambda \in \mathbb{R}^+$ is updated via finding the root of the following one-dimensional monotonic function:
\begin{equation*}
    F(\lambda) \triangleq
    -\Gamma + 
    \frac{\sum_{i = 1}^M \|x_i\|^2 T_i e^{- \lambda \|x_i\|^2}}
    {\sum_{i = 1}^M T_i e^{- \lambda \|x_i\|^2} }.
    \label{update_lam}
\end{equation*}
%
%
In the case of $F(0) \leq 0$, it is observed that the constraint associated with the multiplier $\lambda$ has already been satisfied. 
Therefore, we simply set $\lambda = 0$ without solving the root.
%

\subsubsection{Fix $\bdp$ and update $\bdphi, \widetilde{\bdpsi}, \zeta$}

For fixed $\bdp$, this is equivalent to computing the LM rate (with a prescribed input distribution). 
Similar to \cite{ye2022optimal}\cite{wu2022communication}, we can update the variables $\bdphi$, $\widetilde{\bdpsi}$ according to:
\begin{equation}
    \phi_i = \dfrac{p_i}
    {\sum\limits_{j=1}^N e^{-\zeta d_{ij}}\widetilde{\psi_j} \left(\sum\limits_{k=1}^M s_{kj} p_k \right) }, ~ \widetilde{\psi_j} = \dfrac{1}
    {\sum\limits_{i=1}^M \phi_i e^{-\zeta d_{ij}}}.
    \label{update_dual}
\end{equation}
%
%

Also, we update the variable $\zeta\in\mbbR^{+}$ by finding the root of the monotonic function below:
\begin{equation*}    \label{update_zeta}
    G(\zeta) \triangleq
    \sum_{k=1}^M \sum_{j=1}^N 
    \left[
    \phi_k d_{kj} e^{-\zeta d_{kj}}\widetilde{\psi_j}\left(\sum_{i=1}^M s_{ij} p_i \right) 
    - d_{kj}s_{kj}p_k
    \right].
\end{equation*}
Again, in the case of $G(0) \leq 0$, we can directly set $\zeta = 0$ instead of solving $G(\zeta)=0$. 
%
%

To summarize, we update the variables $\bdp$ and $\bdphi, \widetilde{\bdpsi}, \zeta$ in an alternating manner.
Since this algorithm is based on alternating maximization and double maximization formulation, 
we call it the Alternating Double Maximization (ADM) algorithm.
For clarity, the pseudo-code is presented 
in Algorithm \ref{alg:main}.
The derivation details are shown in the appendix.
\begin{algorithm}[ht]
	
	\renewcommand{\algorithmicrequire}{\textbf{Input:}}
	\renewcommand\algorithmicensure {\textbf{Output:} }
	\caption{Alternating Double Maximization (ADM)}
	\label{alg:main}
	
	\begin{algorithmic}[1]
	\REQUIRE Decoding metric $d_{ij}$, Average power constraint $\Gamma$, Channel transition law $s_{ij}$, Iteration number $max\_iter$ \\
         
        \STATE \textbf{Initialization:} $\bdphi^{(0)} = \bm 1_M$, $\widetilde{\bdpsi}^{(0)} = \bm 1_N$, $\zeta^{(0)}=\lambda^{(0)} = 1$;\\
        
        \FOR{$l = 1 : max\_iter$} 
        \STATE Solve $F(\lambda) = 0$ for $\lambda\in\mathbb{R}^+$ with Newton's method
        \FOR{$i = 1 : M$}
        \STATE Update $p_i ^{(l)}$ according to \eqref{update_p}
        \ENDFOR
        
        \FOR{$i = 1 : M$}
        \STATE Update $\phi_i ^{(l)}\leftarrow\frac{p_i ^{(l)}}{\sum\limits_{j=1}^N e^{-\zeta d_{ij}}\widetilde{\psi_j }^{(l-1)} \left(\sum\limits_{k=1}^M s_{kj} p_k ^{(l)} \right)}$ 
        \ENDFOR
        \FOR{$j = 1 : N$}
        \STATE Update $\widetilde{\psi_j} ^{(l)}\leftarrow\frac{1}{\sum\limits_{i=1}^M \phi_i ^{(l)} e^{-\zeta d_{ij}}}$
        \ENDFOR
        \STATE  Solve $G(\zeta) = 0$ for $\zeta\in\mathbb{R}^+$ with Newton's method
        \ENDFOR
        
        \RETURN $C_\mathrm{LM}$
	\end{algorithmic}
\end{algorithm}


In the following, we give the proof of the convergence in our proposed algorithm.
%
%
This proof relies on the assumption that the optimal points in \eqref{double_max_new} are bounded, a reasonable assumption that holds for the majority of non-extreme cases in practice, and it also ensures the Lipschitz continuity of the gradient. 

\begin{theorem}
    The iteration variables produced by the ADM algorithm converges to the local optimal of the optimized LM rate problem \eqref{LMC_def} satisfying the Nash equilibrium condition.
\end{theorem}
\begin{proof}
    This theorem is an instance of the general theory developed in \cite{yinwotao2013block} regarding the convergence in alternating optimization.
    Noting that the objective in \eqref{double_max_new} is block multi-convex, i.e., convex in each alternating direction, through the dual form \eqref{LM_dual}, and that the gradient in each direction is Lipschitz continuous, the proposed ADM algorithm converges. 
    %
    %
    %
\end{proof}

\section{Numerical Results} 
\label{sec_numerical}

This section evaluates the performance of the ADM algorithm for computing $C_\mathrm{LM}$. 
Our numerical results will show that $C_\mathrm{LM}$ gives a higher achievable rates compared to the LM rate under prescribed channel input probability distributions.
We assess $C_\mathrm{LM}$ over the additive white Gaussian noise (AWGN) channels with IQ imbalances, 
%
modeled by
\vspace{-.02 in}
\begin{equation*}
Y = HX+Z,\quad Z\sim \mathcal{N}({\bf 0},\sigma_Z^2 I).
\end{equation*}
The channel matrix $H \in \mathbb{R}^{2 \times 2}$ is a combination of rotation and scaling effects as
\begin{equation*}
H = 
\begin{pmatrix}
    \eta_1 & 0\\
    0 & \eta_2
\end{pmatrix}
\begin{pmatrix}
    \cos\theta & \sin\theta \\
    -\sin\theta & \cos\theta
\end{pmatrix},
\end{equation*}
in which the parameters $\eta_1$ and $\eta_2$ denote the scaling on the signal,
and the parameter $\theta$ represents the degree of the rotation on the signal. 
We take $\eta_1 = 1$ and $\eta_2 = \eta$ for simplicity.
%
%


%
According to the channel transition law, the channel output alphabet $\mcY$ is $\mathbb{R}^2$. 
But as the AWGN is concentrated around the origin, the channel output $\mcY$ mainly lies within a finite region, and hence
we truncate the alphabet $\mcY$ with a sufficiently large region, e.g. $[-8, 8]\times[-8, 8]$.
%
%
Since our algorithm is applicable to DMC, 
we proceed to discretize the continuous region $[-8,8]\times [-8,8]$ by a set of uniform grid points $\{y_i\}_{i=1}^N$: 
\vspace{-.1 in}
\begin{align*}
    y_{r\sqrt{N}+s} &= (-8+r\Delta y, -8+(s-1)\Delta y),\,\,\Delta y = \frac{16}{\sqrt{N}-1}, \\
    r &=0,1,\cdots ,\sqrt{N}-1,\,\, s = 1,2,\cdots, \sqrt{N}. 
\end{align*}
In this work, we set $N = 10,000$ for the QPSK, 16QAM and 64QAM modulation schemes, and $N = 40,000$ for the 256QAM modulation scheme.

In addition, the distance $d(\bdx,\bdy)=\|\bdy - \hat{H}\bdx\|_2^2$ is used in the decoding metric $q(\bdx, \bdy) = e^{-d(\bdx,\bdy)}$, where $\hat{H}$ is an estimate of the channel matrix $H$. 
In the sequel, we focus on the case where the decoder disregards the issue of mismatch, i.e., $\hat{H} = I$.
Moreover, we define $\mbox{SNR} = {1}/{2\sigma_Z^2}$.
As baseline, we first compute the LM rate $I_\mathrm{LM}$ according to the methodology described in \cite{ye2022optimal}, when all the points in the constellation are equally likely, i.e., having uniform distribution, under the (normalized) average power constraint $\Gamma = 1$. 

Then, with the same position of constellation points, and under the same average power constraint $\Gamma = 1$, we optimize the input distribution $\bm p$ to compute the optimized LM rate $C_{\mathrm{LM}}$ through the ADM algorithm presented in the previous section.

All the experiments are conducted on a PC with 8G RAM and one Intel(R) Core(TM) Gold i5-8265U CPU @1.60GHz. 

\subsection{Convergence behavior}
%
First, let us study the convergence of the proposed ADM algorithm by analyzing the residual errors defined in \eqref{Err_norm}.
\begin{subequations} \label{Err_norm}
\begin{align}
    r_{\phi} &= \sum_{i = 1}^M \left| \phi_i \sum\limits_{j=1}^N e^{-\zeta d_{ij}}\widetilde{\psi_j} \left(\sum\limits_{k=1}^M s_{kj} p_k \right) - p_i \right|,  \\
    r_{\psi} &= \sum_{j = 1}^N \left|
    \left(\widetilde{\psi_j}\sum\limits_{i=1}^M \phi_i e^{-\zeta d_{ij}} -1\right)
    \left(\sum\limits_{k=1}^M s_{kj} p_k \right) \right|, \\
    r_{\zeta} & = \left|G(\zeta)\right|, \quad\quad
    r_{\lambda} = \left|F(\lambda)\right|.
\end{align}
\end{subequations}
%
%

Fig. \ref{fig:modres} shows the convergent trajectories of residual errors versus iteration steps for QPSK, 16QAM, 64QAM and 256QAM, with $(\eta, \theta)= (0.9,{\pi}/{18})$ and SNR = $0$ dB. 
It is noted that the ADM algorithm converges to a tolerance level below $10^{-6}$ for different modulation schemes examined.
\begin{figure}[H]
	\centering
	\includegraphics[width=0.95\linewidth]{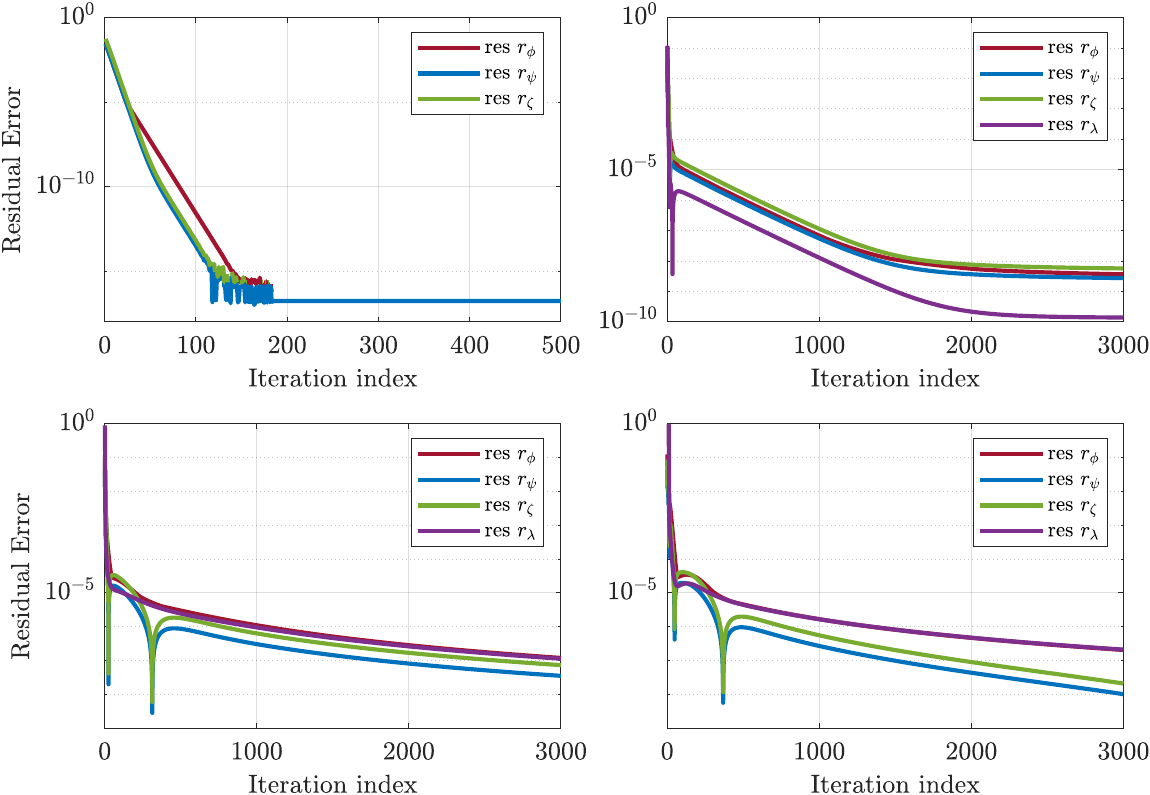}
        \caption{The convergent trajectories of the residual errors $r_{\phi}$ (Red), $r_{{\psi}}$ (Blue), $r_{\zeta}$ (Green), and $r_{\lambda}$ (Purple). 
    Upper Left: The QPSK modulation scheme.
    Upper Right: The 16QAM modulation scheme.
    Lower Left: The 64QAM modulation scheme.
    Lower Right: The 256QAM modulation scheme.}
        \label{fig:modres}
\end{figure}
%

%

%
%
%
%

\subsection{Comparison with the LM Rate}

%
This subsection
compares $C_\mathrm{LM}$ with existing outcomes.
In particular, we compare $C_\mathrm{LM}$ with the LM rate $I_\mathrm{LM}$ computed  
using the method in \cite{ye2022optimal}. 
%
As 
mentioned in the introduction, 
due to the difficulties in computing $C_\mathrm{LM}$, we have not seen 
prior published work with numerical results on it for serving as baseline.
%
%

In this case, the ADM algorithm is run for  $(\eta, \theta)= (0.9,{\pi}/{18}),~(0.8,{\pi}/{18}),~(0.9,{\pi}/{12})$ and $(0.8,{\pi}/{12})$.
The ADM algorithm execution is terminated if the difference between two consecutive iterations is less than $10^{-10}$, or the number of iterations reaches 3000.


\begin{figure}[H]
    \centering
    \includegraphics[width=0.9\linewidth]{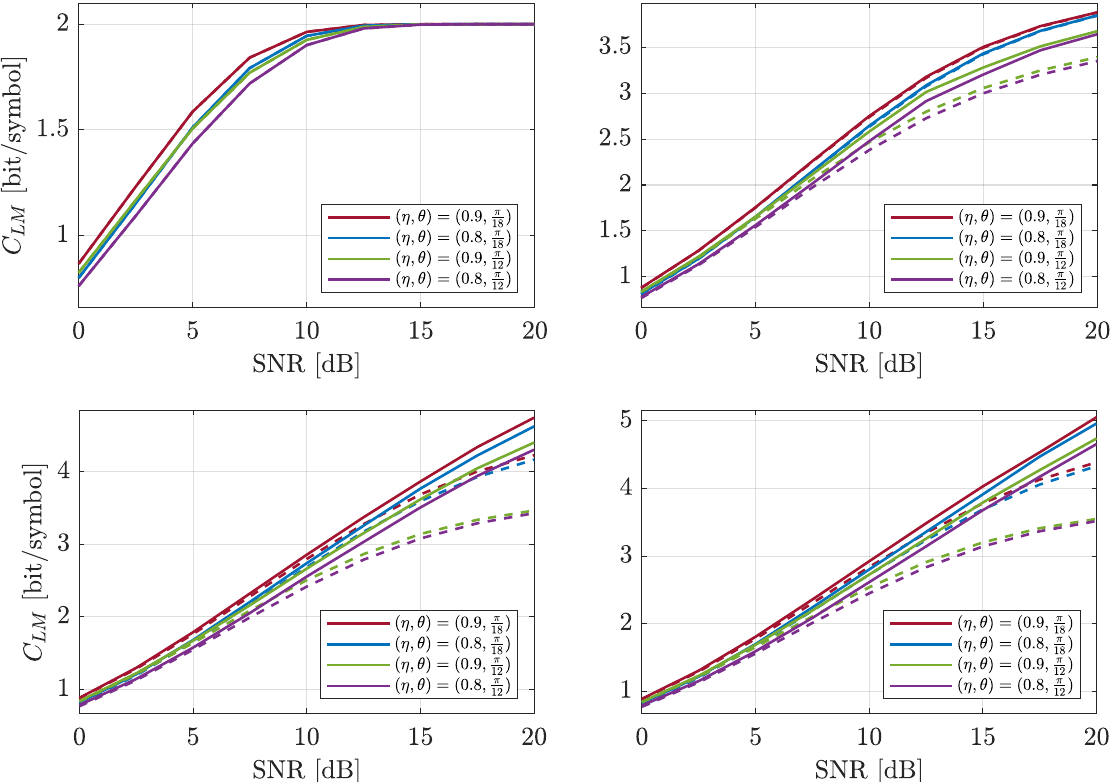}
    \caption{$C_\mathrm{LM}$ (solid) and $I_\mathrm{LM}$ (dashed) versus SNR under different mismatched cases. 
    Upper Left: The QPSK modulation scheme.
    Upper Right: The 16QAM modulation scheme.
    Lower Left: The 64QAM modulation scheme.
    Lower Right: The 256QAM modulation scheme.
    }
    \label{fig:modout}
\end{figure}


Fig. \ref{fig:modout} compares $C_\mathrm{LM}$ with $I_\mathrm{LM}$ for different values of SNR. 
As expected, $C_\mathrm{LM}$ performs consistently larger than $I_\mathrm{LM}$, and the improvement is more noticeable in the high SNR regime. 
Moreover, for a fixed SNR, the 
gap between $C_\mathrm{LM}$ and $I_\mathrm{LM}$
becomes more remarkable with increasing modulation orders.
In addition, it is demonstrated that both $C_\mathrm{LM}$ and $I_\mathrm{LM}$ drop with decreasing values of  $\eta$ or increasing values of $\theta$.

%
%
%

%
%

%

\section{Conclusion} \label{sec_conclusion}
%
This work presented an approach for evaluating $C_\mathrm{LM}$ in mismatched decoding. 
The main difficulty of this problem was that the corresponding optimization was non-convex when the channel input probability distribution needed to be optimized. 
%
%
To tackle this difficulty, we proposed a novel dual form of the LM rate, and transformed the non-convex max-min problem to a double maximization, whose each direction was convex. 
This new formulation led to a maximization problem setup, and then an alternating double maximization algorithm was presented, where each step only needed to solve a closed-form iteration. 
Numerical results demonstrated evident improvement of $C_\mathrm{LM}$ over the LM rate without input optimization for different modulation schemes, and validated the convergence of the proposed algorithm.  
Due to limited space, discussions on complexity and real-world applications are deferred to a subsequent long article.

\bibliographystyle{bibliography/IEEEtran}
\bibliography{bibliography/CLM_REF}

\newpage
\quad
\newpage
\begin{appendix}





%
\subsection{Proof of Proposition \ref{prop:equi}} \label{ap.B}
%
We will elucidate the equivalence between our proposed dual form of LM rate \eqref{LM_dual} and the one introduced in \cite[p. 19]{Scarlett2020Information},

\begin{proof}
As it shown in \cite[p. 19]{Scarlett2020Information}, the dual form of the LM rate is given by the following maximization problem:
\begin{equation*} \label{Scarlett_dual_LM_1}
    \max _{\zeta \geq 0, a(\cdot)} \sum_{x, y} Q_X(x) W(y \mid x) \log \frac{q(x, y)^{\zeta} e^{a(x)}}{\sum_{\bar{x}} Q_X(\bar{x}) q(\bar{x}, y)^{\zeta} e^{a(\bar{x})}}.
\end{equation*}
Using the notations in our paper, this dual form is written as: 
\begin{equation} \label{Scarlett_dual_LM_2}
    \max_{\zeta,\widehat{\bdphi}}~~ \sum_{i=1}^{M}\sum_{j=1}^{N} p_{i} s_{ij} \log\frac{e^{-\zeta d_{ij}} \widehat{\phi_{i}}}{\sum_{k=1}^{M} p_{k} e^{-\zeta d_{kj}}\widehat{\phi_{k}}},
\end{equation}
where $\widehat{\phi_{k}} = e^{a(x_{k})}$.

On the other hand, 
recall that the following function $g_{\mathrm{LM}}(\bdphi,\bdpsi,\zeta)$ is the objective function of our proposed dual form of LM rate \eqref{LM_dual}
%
\begin{align*}
    & g_{\mathrm{LM}} (\bdphi ,\bdpsi ,\zeta) = 1-\sum_{i=1}^M \sum_{j=1}^N \phi_i e^{-\zeta d_{ij}} \psi_j - \zeta \sum_{i=1}^M \sum_{j=1}^N d_{ij}s_{ij}p_i\\
    & - \sum_{i=1}^M p_i \log p_i - \sum_{j=1}^N q_j \log q_j  + \sum_{i=1}^M p_i \log \phi_i+ \sum_{j=1}^N q_j \log \psi_j.
\end{align*}
%

Taking the derivative of $g_{\mathrm{LM}}(\bdphi,\bdpsi,\zeta)$ with respect to $\bdpsi$ leads to the condition
\begin{equation*}
    \psi_{j}^{*} = \frac{q_{j}}{\sum_{k=1}^{M} e^{-\zeta d_{kj}}\phi_{k}}.
\end{equation*}
Then, substituting this condition into $g_{\mathrm{LM}}(\bdphi,\bdpsi,\zeta)$, we have
\begin{equation*}
\begin{aligned}
    & g_{\mathrm{LM}}(\bdphi,\bdpsi^{*},\zeta) = \sum_{i=1}^{M} p_i \log \phi_i - \sum_{i=1}^{M} p_i \log p_i \\
    & - \sum_{j=1}^N q_j \log\left(\sum_{k=1}^{M} e^{-\zeta d_{kj}}\phi_{k}\right) -\zeta \sum_{i=1}^M \sum_{j=1}^N d_{ij}s_{ij}p_i.
\end{aligned}
\end{equation*}
Furthermore, denoting $\widehat{\phi_{i}} = \frac{\phi_{i}}{p_{i}}$, we have 
\begin{equation} \label{dual_LM_e}
\begin{aligned}
    g_{\mathrm{LM}}&(\bdphi,\bdpsi^{*},\zeta) = \sum_{i=1}^{M} p_i \log \widehat{\phi_i}+\sum_{i=1}^M \sum_{j=1}^N (-\zeta d_{ij})s_{ij}p_i \\
    &- \sum_{j=1}^N \left(\sum_{i=1}^M p_i s_{ij}\right) \log\left(\sum_{k=1}^{M} p_{k} e^{-\zeta d_{kj}} \widehat{\phi_{k}}\right).
\end{aligned}
\end{equation}
This is exactly the same as the objective in \eqref{Scarlett_dual_LM_2} by summing the three terms and extracting the coefficients $p_{i}s_{ij}$.
\end{proof}
\vspace{-.1in}

\subsection{Derivation Details of \eqref{update_p} and $F(\lambda)$} \label{ap.D}
The primary derivation process of the algorithm will be presented below. 
%

Firstly, recall the optimization problem \eqref{double_max_new} under consideration,
\begin{equation}
    \begin{aligned}
        \max_{\substack{\bdp, \bdphi, \widetilde{\bdpsi}, \\ \zeta \geq 0}} \quad & - \sum_{i=1}^M p_i \log p_i + \sum_{i = 1}^M p_i \log T_i(\bdphi, \widetilde{\bdpsi}, \zeta) + 1 \\
        \mathrm{s.t.}\quad &  \sum_{i = 1}^M p_i  = 1, \quad \sum_{i = 1}^M p_i \|x_i\|^2 \leq \Gamma, 
    \end{aligned}
    \label{a_double_max_new}
\end{equation}
where the function $T_i(\bdphi, \widetilde{\bdpsi}, \zeta)$ is defined as
\begin{footnotesize}
\begin{equation*}
    T_i(\bdphi, \widetilde{\bdpsi}, \zeta) 
    \triangleq \phi_i \exp
    \left(\sum\limits_{j=1}^N s_{ij}
    \left[  \!-\!\!\!\sum\limits_{k=1}^M\!\! \phi_k e^{-\!\zeta d_{kj}}\widetilde{\psi_j}
    \!\!+\!\! \log\widetilde{\psi_j}
    \!\!-\!\!  \zeta d_{ij}  \right]
    \right).
\end{equation*}
\end{footnotesize}

Here we consider the optimization problem with respect to the input distribution $\bdp$.

Under fixed $\bdphi, \bdpsi, \zeta$, the  Lagrangian 
is given by: 
\begin{equation}
    \begin{aligned}
        \mathcal{L}(\bdp; & \lambda, \eta) = 
         - \sum_{i=1}^M p_i \log p_i + \sum_{i = 1}^M p_i \log T_i  +1 \\
        & - \lambda\left( \sum_{i = 1}^M p_i \|x_i\|^2 - \Gamma\right)
        - \eta\left( \sum_{i = 1}^M p_i - 1 \right),
    \end{aligned}
    \label{a_lmc_lagarange}
\end{equation}
where $\lambda \in \mathbb{R}^+$ and $\eta \in \mathbb{R}$ are dual variables.
Take the partial derivative of the Lagrangian $\mathcal{L}(\bdp; \lambda, \eta)$ with respect to $\bdp$, and we could obtain
\begin{equation}
    \frac{\partial \mathcal{L}}{\partial p_i} = - 1 - \log p_i + \log T_i - \lambda \|x_i\|^2 - \eta = 0.
\end{equation}
Hence, the optimal solution is given by
\begin{equation}    \label{pppp}
    p_i^{*} = T_i e^{- \lambda \|x_i\|^2} e^{-1 - \eta}.
\end{equation}
By substituting $p_i^*$ into the Lagrangian \eqref{a_lmc_lagarange}, the dual form of \eqref{a_double_max_new} can be written as:
\begin{equation}
    \max_{ \lambda \geq 0, \,\, \eta} \,\,
    f_0(\lambda,\eta) \triangleq
    - e^{-1-\eta}
    \sum_{i = 1}^M T_i e^{- \lambda \|x_i\|^2}
    -1
    - \lambda\Gamma - \eta .
    \label{eq:dual21}
\end{equation}

To solve \eqref{eq:dual21}, we firstly maximize it in the direction $\eta$ with a fixed $\lambda$, and then maximize it in the direction $\lambda$.

More specifically, taking the partial derivative of $f_0$ with respect to $\eta$, we have
\begin{equation*}
    \frac{\partial f_0}{\partial \eta} 
    = e^{-1-\eta} \sum_{i = 1}^M T_i e^{- \lambda \|x_i\|^2} - 1
     = 0.
\end{equation*}
Then we substitute its solution
$$
\eta^* = -1 + \log\left( \sum_{i = 1}^M T_i e^{- \lambda \|x_i\|^2} \right)
$$ 
into the objective function \eqref{eq:dual21}, and obtain an optimization problem with respect to the variable $\lambda$ only:
\begin{equation}
    \max_{\lambda \geq 0} \,\,
    f_1(\lambda) \triangleq
    -1 - \lambda\Gamma - \log\left( \sum_{i = 1}^M T_i e^{- \lambda \|x_i\|^2} \right).
    \label{eq:dual22}
\end{equation}
Noting that $f_1$ is a concave function, we can solve \eqref{eq:dual22} using the Newton's method with great efficiency. 

In detail, denoting $F(\lambda) \triangleq f_1^{'}(\lambda)$, we obtain 
\begin{equation*}
    F(\lambda) =
    -\Gamma + 
    \frac{\sum_{i = 1}^M \|x_i\|^2 T_i e^{- \lambda \|x_i\|^2}}
    {\sum_{i = 1}^M T_i e^{- \lambda \|x_i\|^2} } 
    \triangleq
    -\Gamma - \frac{f_2^{'}(\lambda)}{f_2(\lambda)},
\end{equation*}
which is actually \eqref{update_lam} above.
Its derivative is negative, due to the Cauchy-Schwarz inequality, i.e., 
\begin{equation*}
\begin{aligned}
    F^{'}(\lambda) = -\frac{f_2^{''}(\lambda) f_2(\lambda) - \left( f_2^{'}(\lambda) \right) ^2}{\left( f_2(\lambda) \right) ^2} 
    \leq 0.
\end{aligned}
\end{equation*}
Then we can solve $F(\lambda) = 0$ by using the Newton's method with great efficiency, and update $\lambda$ according to the solution.

Finally, considering $\sum_{i=1}^M p_i=1$ and \eqref{pppp}  jointly, we can update $p_i$ by
\begin{equation*}
    p_i = \frac{p_i^*}{\sum_{i^{'} = 1}^M p_{i^{'}}^*}
    =  \frac{T_i e^{- \lambda \|x_i\|^2}}{\sum_{i^{'} = 1}^M T_{i^{'}} e^{- \lambda \|x_{i^{'}}\|^2}}, 
    \label{a_update_p}
\end{equation*}
which is actually \eqref{update_p} above.

\subsection{Derivation Details of \eqref{update_dual} and $G(\zeta)$} \label{ap.F}
%
Here we give derivation details in the case of a fixed input distribution  $\bdp$ . 

We aim to maximize the expression \eqref{double_max_new} by formulating an optimization problem with respect to the variables $\bdphi$, $\widetilde{\bdpsi}$, and $\zeta$, 
thereby actually calculating the LM rate.

Taking the partial derivative of the objective function in \eqref{a_double_max_new} with respect to variables $\bdphi, \widetilde{\bdpsi}$, we have
\begin{equation*}
\begin{aligned}
    &-\sum\limits_{j=1}^N e^{-\zeta d_{ij}}\widetilde{\psi_j} \left(\sum_{k=1}^M s_{kj} p_k \right) 
    + \frac{p_i}{\phi_i} = 0,
    \\
    &-\sum\limits_{i=1}^M \phi_i e^{-\zeta d_{ij}} \left(\sum_{k=1}^M s_{kj} p_k \right) 
    + \frac{\left(\sum\limits_{k=1}^M s_{kj} p_k \right)}{\widetilde{\psi_j}} = 0.
\end{aligned}
\end{equation*}
Accordingly, we update $\bdphi, \widetilde{\bdpsi}$ by
\begin{equation*}
    \phi_i = \dfrac{p_i}
    {\sum\limits_{j=1}^N e^{-\zeta d_{ij}}\widetilde{\psi_j} \left(\sum\limits_{k=1}^M s_{kj} p_k \right) }, ~ \widetilde{\psi_j} = \dfrac{1}
    {\sum\limits_{i=1}^M \phi_i e^{-\zeta d_{ij}}},
\end{equation*}
which is actually \eqref{update_dual} above.

Next, take the partial derivative of the objective function with respect to $\zeta$, i.e.,
%
%
\begin{equation*}
    G(\zeta) \triangleq
    \sum_{k=1}^M \sum_{j=1}^N 
    \left[
    \phi_k d_{kj} e^{-\zeta d_{kj}}\widetilde{\psi_j}\left(\sum_{i=1}^M s_{ij} p_i \right) 
    - d_{kj}s_{kj}p_k
    \right],
\end{equation*}
which is actually \eqref{update_zeta} above.
Noting that
\begin{equation*}
    G^{'}(\zeta) = 
    -\sum_{k=1}^M \sum_{j=1}^N \phi_k d_{kj}^2 e^{-\zeta d_{kj}}\widetilde{\psi_j}\left(\sum_{i=1}^M s_{ij} p_i \right) 
    \leq 0,
\end{equation*}
$G(\zeta)$ is monotonic. Hence, we can similarly find the root of $G(\zeta)$ using the Newton's method with great efficiency, and update $\zeta$ according to the solution.

At this point, the derivations of the ADM algorithm are fully presented.



\end{appendix}

\end{document}